\documentclass[11pt,letterpaper]{article}

\usepackage[utf8]{inputenc}
\usepackage[T1]{fontenc}
\usepackage{amsmath,amssymb,amsthm,amsfonts,mathrsfs,bm}

\usepackage[margin=1in]{geometry}

\usepackage{graphicx}
\usepackage{booktabs}
\usepackage{float}
\usepackage{array}
\usepackage{xcolor}
\usepackage{tikz}
\usepackage{enumerate}
\usepackage[linesnumbered,ruled,vlined]{algorithm2e}
\usepackage[hidelinks]{hyperref}

\newtheorem{theorem}{Theorem}[section]
\newtheorem{lemma}[theorem]{Lemma}

\newtheorem{corollary}[theorem]{Corollary}

\newcommand{\alg}{\mathcal{A}}
\newcommand{\pattern}{\mathcal{P}}
\newcommand{\cycf}{\mathcal{F}_{\mathrm{cyc}}}

\SetAlgorithmName{Algorithm}{Algorithm}{List of Algorithms}
\title{Communication-Oblivious Permutation and Sorting}
\author{
  Qi Dong\textsuperscript{1}\quad
  Yilei Wang\textsuperscript{2}\quad
  Ke Yi\textsuperscript{1}\thanks{Corresponding author.}\\[0.6em]
  \small\textsuperscript{1}The Hong Kong University of Science and Technology, Hong Kong SAR\\
  \small\textsuperscript{2}ByteDance, Hangzhou, China\\[0.3em]
  \small\href{mailto:qdongaf@connect.ust.hk}{\nolinkurl{qdongaf@connect.ust.hk}}\quad
  \href{mailto:wangyilei.wyl@bytedance.com}{\nolinkurl{wangyilei.wyl@bytedance.com}}\quad
  \href{mailto:yike@cse.ust.hk}{\nolinkurl{yike@cse.ust.hk}}
}
\date{}
\hypersetup{
  pdftitle={Communication-Oblivious Permutation and Sorting},
  pdfauthor={Qi Dong, Yilei Wang, Ke Yi}
}

\begin{document}

\maketitle

\begin{abstract}
With the help of secure hardware such as trusted execution environments (TEEs), private data can be uploaded and computed in untrusted cloud clusters, so that an adversary controlling the cloud can only see the data in encrypted form.  However, the communication pattern between the processors may still reveal sensitive information about the data, so it is important to develop parallel algorithms in the BSP model whose communication pattern is independent of the input data, a property known as \textit{communication obliviousness}. In this paper, we study such algorithms for three fundamental problems: permutation, sorting, and 0-1 sorting. For the first two problems, we prove a tight (up to lower-order terms) communication complexity of $(2-2/p)n\pm o(n)$.  Here, $p$ is the number of processors and $n$ is the input size, assuming $2\le p$ and $p=o(n^{1/3})$.  Our lower bound holds regardless of the number of rounds and local memory size, while for balanced permutation and sorting the upper bound is achieved in a constant number of rounds with load at most $n/p+o(n/p)$.  For 0-1 sorting, we also derive a tight communication complexity of $C_{\mathrm{cut}}(p) \cdot n \pm o(n)$, where $C_{\mathrm{cut}}(p)$ is the optimal value of a linear program parameterized by $p$. Although there appears to be no closed-form for $C_{\mathrm{cut}}(p)$, we show that it is  strictly smaller than $2-2/p$ for all $p\ge 3$, thus separating 0-1 sorting from general sorting in terms of their communication complexity.

\end{abstract}

\section{Introduction}

Large-scale analytics increasingly relies on cloud clusters because the underlying datasets and computations exceed the capacity of a single machine.  A standard model for computation using a cluster is the Bulk Synchronous Processing (BSP) model \cite{10.1145/79173.79181}, also known as the Massively Parallel Computation (MPC) model \cite{10.1145/2594538.2594558,mpc4mpc}.  In this model, which will be defined more precisely in Section \ref{sec:preliminaries}, there are $p$ processors, each holding $n/p$ data initially and (ideally) also in the end. Computation proceeds in synchronous \textit{rounds}, where each round consists of a local computation stage and a communication stage. Good BSP algorithms should have a small number of rounds, low total communication cost, and low memory requirement for each processor.  For example, the classical sample sort algorithm runs in 3 rounds, has a communication cost of $n + o(n)$, and requires $n/p + o(n/p)$ local memory per processor, assuming the parameter range $2\le p \le \tilde{o}(\sqrt{n})$\footnote{The  $\tilde{ }$  notation hides logarithmic factors.} \cite{regularSampleSort}.  Note that both its communication cost and local memory requirement are clearly optimal up to lower-order terms.

Datasets uploaded to a cloud often contain sensitive medical, financial, behavioral, or organizational information. While the data can be encrypted during communication and local computation,
the adversary can still infer sensitive information from the communication pattern, which is known as side-channel attacks \cite{7163052,10.1145/2810103.2813695}.  Most cloud-based secure distributed computation systems suffer from such attacks as they are based on Trusted Execution Environments (TEEs), which prevent the untrusted cloud from seeing the data in plaintext but still reveal communication patterns.  Thus, it is important to design algorithms whose communication pattern is independent of the input data, a property known as \textit{communication obliviousness}.  The sample sort algorithm, however, is not communication-oblivious (henceforth simply called oblivious): The communication pattern reveals which processors originally had the smallest $n/p$ items, the second smallest $n/p$ items, etc.  A trivial oblivious algorithm sends all data to one processor, sorts them locally, and distributes them back in sorted order.  This algorithm incurs $(2-2/p)n$ communication, (almost) twice as much as non-oblivious sort, and more critically, requires $n$ local memory, forfeiting all parallelism.  Better parallel oblivious sorting algorithms have been designed \cite{columnsort,Ngai} using $n/p+o(n/p)$ local memory, but the best communication cost they can achieve is $3n$.

We use \emph{load}, the maximum number of words received by a processor in a round, to measure per-processor communication. To keep this communication proportional to the local data size, we seek algorithms with $O(n/p)$ load and a small constant factor. For balanced permutation, sorting, and 0–1 sorting, our algorithms achieve load at most $n/p+o(n/p)$ while attaining the optimal leading communication cost.

\subsection{Our Results}

We study oblivious BSP algorithms for three fundamental problems: permutation, sorting, and 0-1 sorting.  Both permutation and 0-1 sorting can be considered as special cases of sorting, where the keys of the items are drawn from $\{1,\dots,p\}$ and $\{0,1\}$, respectively.  For general sorting, the keys are taken from a totally ordered domain and we are only allowed to use comparisons to deduce their ordering.

\paragraph{Permutation and sorting.}
We first prove a \((2-2/p)n\) lower bound on the communication cost for oblivious permutation, regardless of the number of rounds and local memory size.  The same lower bound thus also holds for sorting, meaning that the trivial algorithm above that sends all data to one processor is already optimal in terms of the communication complexity.  Then we design better parallel oblivious permutation and sorting algorithms with a communication cost of $(2-2/p)n+o(n)$, assuming $2\le p$ and $p=o(n^{1/3})$.  For balanced permutation and sorting, these algorithms have load at most $n/p+o(n/p)$.  The communication cost is thus optimal up to lower order terms (even for a constant $p$).  The permutation algorithm is deterministic and perfectly communication-oblivious, while the sorting algorithm is randomized with negligible failure and non-oblivious probability.  The lower bound holds for deterministic algorithms, but can be extended to randomized algorithms (with negligible failure and non-oblivious probability) by just adding an $-o(n)$ term.

\paragraph{0-1 sorting.}
The situation of oblivious 0-1 sorting is more intricate.  We prove a tight communication complexity of \(C_{\mathrm{cut}}(p)\cdot n \pm o(n)\), where $C_{\mathrm{cut}}(p)$ is the optimum of a linear program (LP) parameterized by $p$. The lower bound again holds for both deterministic and randomized algorithms, regardless of the number of rounds and local memory size.  Our algorithm achieving the upper bound, under the regime $2\le p$ and $p=o(n^{1/3})$, has load $o(n/p)$ and uses $p\cdot\omega(1)$ rounds.  More generally, there is a trade-off between the number of rounds and the communication cost; please see Theorem \ref{thm:01-original-coefficient} for the precise statement.

Finally, although we do not have a closed-form for $C_{\mathrm{cut}}(p)$, we are able to show that
\[
\frac1p\left\lfloor\frac{4p-2}{3}\right\rfloor
\le C_{\mathrm{cut}}(p)
\le
\min\left\{
\frac32,\,
2-\frac2p-\frac16
\right\},
\]
for all \(p\ge3\). Consequently, for every \(p\ge5\), the leading communication coefficient of oblivious 0-1 sorting is strictly between those of non-oblivious sorting and oblivious general sorting.

\subsection{Related Work}

\paragraph{Communication-oblivious BSP algorithms.}
BSP/MPC algorithms have been extensively studied in both the algorithms and the systems literature, as the model resembles many big-data computing platforms such as MapReduce and Spark \cite{10.5555/1873601.1873677,10.1145/1327452.1327492,10.5555/2228298.2228301}.  There are generic techniques \cite{mpc4mpc,tao2024parallel} that can convert BSP algorithms to oblivious ones with constant-factor overhead under suitable parameter assumptions. Applying either compiler to the one-round non-oblivious balanced permutation algorithm gives roughly $4n$ communication, which does not attain the optimal $(2-2/p)n+o(n)$ communication cost obtained here.

\paragraph{RAM-oblivious algorithms.}
In addition to the communication pattern, side-channel information also includes the memory access pattern.  To differentiate, we refer to the latter as RAM-oblivious algorithms.  While we only focus on communication obliviousness in this paper, one can combine our algorithms with RAM-oblivious algorithms for the local computation part to achieve full obliviousness.

RAM-oblivious algorithms enjoy a longer history.  ORAM is a generic technique that can make any algorithm RAM-oblivious by increasing its running time by an $O(\log n)$ factor \cite{10.1145/233551.233553,optorama2020}.  Alternatively, custom RAM-oblivious algorithms have been designed for many problems, including permutation, sorting, and 0-1 sorting (a.k.a. compaction), with the former two in $\Theta(n\log n)$ time \cite{benes1964permutation,10.1145/800061.808726} and the latter in $\Theta(n)$ time \cite{asharov2020tightcompaction,dittmer2020tightcompaction}.  However, the hidden constants are still very large and far from optimal: The best oblivious sorting algorithm is randomized shell sort \cite{goodrich2011randomized}, which performs $24\cdot n\log n$ comparisons. This also motivates our assumption $p=o(n^c)$: when $p$ is close to $n$ and local memory is limited, communication-oblivious algorithms resemble RAM-oblivious algorithms, with communication cost corresponding to memory-access cost.

For example, consider the realistic model with $n=p$ and $O(1)$ total local memory per processor. Each processor's state fits in a constant-size RAM block, and each word communicated can be simulated with $O(1)$ RAM accesses. Conversely, a linear-space RAM algorithm with constant private workspace can be simulated by distributing its memory across the processors. These simulations preserve fixed oblivious patterns with constant overhead in data movement. For indivisible records, the resulting optimal deterministic oblivious communication bounds are $\Theta(n\log n)$ for balanced permutation and general sorting, matching their RAM-access bounds~\cite{benes1964permutation,10.1145/800061.808726}, and $\Theta(n)$ for unstable 0-1 sorting~\cite{asharov2020tightcompaction,dittmer2020tightcompaction}.

\paragraph{Distributed oblivious analytics.}
Systems such as Opaque, SODA, and Jodes motivate commu\-nication-efficient primitives for oblivious distributed analytics~\cite{opaque,soda,jodes2020}. Our algorithms provide building blocks for this setting when combined with RAM-oblivious local computation. Ngai et al.\ develop oblivious sorting and shuffling across multiple enclaves~\cite{Ngai}. Our work establishes tight leading-order upper and lower bounds on total communication for permutation, sorting, and 0-1 sorting.

\section{Preliminaries}
\label{sec:preliminaries}

\subsection{Model of Computation}

Throughout this work, $n$ denotes the number of input items and $p\ge2$ the number of processors. We assume that $p\mid n$. %
For a nonnegative integer $m$, write $[m]:=\{1,\ldots,m\}$, with $[0]:=\emptyset$.
We consider the bulk-synchronous parallel (BSP) model~\cite{10.1145/79173.79181}, in which the processors collaboratively execute an algorithm $\alg$ on input $X=(X_1,\ldots,X_p)$, where processor $i$ initially holds $X_i$. We assume $|X_i|=n/p$ for every $i\in[p]$. For all the problems considered in this paper, the output $\alg(X)$ is a permutation of $X$.

The computation may take multiple rounds. Let $X^r=(X_1^r,\ldots,X_p^r)$ denote the data distribution at the beginning of round $r$; in particular,  $X_i^1:=X_i$ for every $i\in[p]$.
In round $r$, processor $i$
sends $Y_{i\to j}^r \subseteq X_i^r$ to processor $j$ for $j\ne i$, and retains $Y_{i\to i}^r \subseteq X_i^r$ locally. The items held by processor $j$ at the beginning of round $r+1$ are therefore given by $X_j^{r+1}=\bigcup_{i\in[p]}Y_{i\to j}^r$, for every $j\in[p]$.
The final output is $\alg(X):=(X_1^{R+1},\ldots,X_p^{R+1})$.

We consider two versions of the model. In the \textit{ideal} model, the algorithm is given complete knowledge of the data distribution; equivalently, an oracle provides the processors with routing instructions at no communication cost.
In the \textit{realistic} model, processors make routing decisions locally and may need to receive auxiliary information from other processors.
We measure communication in words, assuming that each input item fits in one $O(\log n)$-bit word. All our lower bounds hold in the ideal model, whereas the upper bounds hold in the realistic model.

The \emph{communication pattern} of $\alg$ on input $X$ is $\pattern_{\alg}(X):=(L^1,\dots,L^R)$, where $L^r=(L_{ij}^r)_{i,j\in[p]}$ and $L_{ij}^r$ is the total length, in words, of the message from processor $i$ to processor $j$ in round $r$, including the items in $Y_{i\to j}^r$ and auxiliary words, possibly padded with dummy tuples for achieving obliviousness; for every $i\in[p]$, we set $L_{ii}^r=0$.
Algorithms may be randomized, in which case $\pattern_{\alg}(X)$ is a random variable induced by the algorithm's internal randomness.

Algorithms in the BSP model are evaluated by three measures. The first is the number of rounds $R$, ideally a constant.
For a fixed pattern $\pattern=(L^1,\ldots,L^R)$, define its communication cost by $C(\pattern):=\sum_{r,i,j}L_{ij}^r$.
The second measure is the communication cost $C(\alg):=\max_X\mathbb E[C(\pattern_{\alg}(X))]$, where the expectation is over the algorithm's internal randomness.
The third is the load, defined by $\operatorname{load}(\alg):=\max_{X,\,r\in[R],\,j\in[p]}\sum_{i\ne j}L_{ij}^r$, where the maximum also ranges over internal random choices.

\paragraph{Communication obliviousness.}
For any problem in which $\alg(X)$ is a permutation of $X$, there is a trivial algorithm with communication cost at most $n$ in the ideal model: simply send each item to its target processor. Its communication pattern, however, depends on the input. We are interested in algorithms whose communication patterns are independent of, or \textit{oblivious} to, the input. Formally, let $\Phi(X)$ be the \emph{leakage function}, which specifies the information that an adversary is allowed to observe. An algorithm~$\alg$ is \emph{$\delta$-statistically communication-oblivious} (or simply $\delta$-oblivious) with respect to~$\Phi$ if, for every pair of inputs $X,X'$ such that $\Phi(X)=\Phi(X')$, we have $\Delta\bigl(\pattern_{\alg}(X),\,\pattern_{\alg}(X')\bigr)
\;\le\;
\delta$, where $\Delta(\mathsf{D}_1,\mathsf{D}_2):=
\frac{1}{2}\sum_{s}\bigl|\Pr[\mathsf{D}_1=s]-\Pr[\mathsf{D}_2=s]\bigr|$ is the statistical distance between the distributions of the random variables $\mathsf{D}_1$ and~$\mathsf{D}_2$. We also call $\delta$ the non-oblivious probability. A function \(\delta=\delta(n)\) is \emph{negligible} if \(\delta=n^{-\omega(1)}\), i.e., if for every fixed \(c>0\), it satisfies \(\delta\le n^{-c}\) for all sufficiently large \(n\). An algorithm with $\delta=0$ is \emph{perfectly oblivious}; note that a deterministic oblivious algorithm must be perfectly oblivious.

\subsection{Problems}

This paper investigates the following three problems. In each problem, $t_x$ denotes the key of item $x$.

\paragraph{Permutation.}
In this problem, each item $x$ is given a key $t_x\in[p]$, and the algorithm is required to place $x$ on processor $t_x$. The leakage function is the destination-count vector \(\Phi(X)=D=(d_1,\ldots,d_p)\), where \(d_i:=\sum_{j\in[p]}|\{x\in X_j:t_x=i\}|\). We allow the final item counts to be revealed, while hiding the assignment of individual items beyond these counts.
If only $n$ were revealed, any perfectly oblivious algorithm would require every processor to receive at least $n-n/p$ words in total, since any processor could be the destination of all $n$ items. This forces at least $(p-1)n$ total communication, attained in one round by routing items directly to their destinations and padding every message between distinct processors to $n/p$ words.
A special case, which we call \textit{balanced permutation}, has $d_i=n/p$ for every $i\in[p]$. Since $D$ is then determined by $n$ and $p$, its leakage function reduces to $\Phi(X)=n$.

\paragraph{Sorting.}
Each item $x$ is given a unique key $t_x$ from a totally ordered domain, and the algorithm must place the items in nondecreasing key order: for every $1\le i<j\le p$, all keys on processor $i$ are no greater than those on processor $j$. We further require each processor to hold exactly $n/p$ items upon termination. The leakage function is $\Phi(X):=n$. It should be clear that, in the ideal model, sorting is equivalent to balanced permutation, since the oracle knows the destination of each item.

\paragraph{0-1 sorting.}
The 0-1 sorting problem is the special case of sorting where every item $x$ has key $t_x\in\{0,1\}$. Note that ties can be broken arbitrarily, i.e., the sorting is allowed to be unstable.  As will be seen later, our permutation lower bound also holds for stable 0-1 sorting, so stableness makes 0-1 sorting as hard as general sorting. Its leakage function is likewise $\Phi(X):=n$.

\subsection{General Reductions}
\label{sec:general-reductions}

We state two general reductions that simplify the classes of algorithms we must consider. They adapt standard arguments from the oblivious-algorithms literature~\cite{foundationDO} to our setting.

The first shows that there is little difference between Monte Carlo and Las Vegas algorithms in our setting.

\begin{lemma}[Failure removal]\label{lem:failure_removal}
Suppose there is an $R$-round, $\delta$-oblivious algorithm $\alg$ that may fail with probability at most $\delta'$ and outputs a failure flag upon failure. Let $\alg_{\mathrm{fb}}$ be an always-correct non-oblivious algorithm using at most $R_{\mathrm{fb}}$ rounds. Then there is an always-correct, $(\delta+2\delta')$-oblivious algorithm using at most $R+R_{\mathrm{fb}}$ rounds and having communication cost at most $C(\alg)+\delta' C(\alg_{\mathrm{fb}})$. This holds in both the realistic and ideal models.
\end{lemma}

\begin{proof}
We simply run $\alg$ first and invoke $\alg_{\mathrm{fb}}$ when $\alg$ announces failure; call the resulting algorithm $\alg'$. The round and communication bounds are immediate; we only need to show that $\alg'$ is $(\delta+2\delta')$-oblivious. Fix two inputs $X,X'$ with $\Phi(X)=\Phi(X')$. Let $\pattern_{\alg}(X)$ and $\pattern_{\alg}(X')$ be the communication-pattern distributions, and let $\pattern_{\alg'}(X)$ and $\pattern_{\alg'}(X')$ be the communication-pattern distributions of the wrapped algorithm. On the event that $\alg$ does not fail, the wrapped communication pattern is a deterministic postprocessing of the original communication pattern. On the failure event, whose probability is at most $\delta'$ for each input, the fallback communication pattern may be arbitrary. Coupling the non-failure parts and union-bounding the two failure events gives
\[
\Delta(\pattern_{\alg'}(X),\pattern_{\alg'}(X'))
\le
\Delta(\pattern_{\alg}(X),\pattern_{\alg}(X'))+\Pr[\alg\text{ fails on }X]+\Pr[\alg\text{ fails on }X']
\le \delta+2\delta'.\qedhere
\]
\end{proof}

For all three problems, processors can detect failure with $O(p)$ communication cost, and there are always-correct fallback algorithms using $O(1)$ rounds and $O(n)$ communication in the parameter regimes considered below. Sorting and 0-1 sorting have load at most $n/p+o(n/p)$ via sample sort~\cite{regularSampleSort}, while direct routing for permutation delivers at most $\max_j d_j$ words to any processor on an input with destination counts $D$. Thus, when $\delta'$ is negligible, failure removal adds only $o(n)$ expected communication and $O(1)$ rounds, with load bounded by the maximum of the original and fallback loads.

The second reduction shows that derandomization is possible in the ideal model when the inputs are restricted to a small family.
\begin{lemma}[Derandomization in the ideal model]\label{lem:finite_family_derand}
Consider the ideal model. Suppose there is a Las Vegas, $\delta$-oblivious algorithm $\alg$ for a finite family $\mathcal{F}$ of inputs with the same leakage value.
If $(|\mathcal{F}|-1)\delta<1$, then there exists a deterministic algorithm $\alg^\star$ with communication cost at most $\frac{C(\alg)}{1-(|\mathcal{F}|-1)\delta}$.
\end{lemma}

\begin{proof}
Pick any $X_0\in\mathcal{F}$ and sample a communication pattern $\widetilde{\pattern}$ from $\pattern_{\alg}(X_0)$. By definition, $\mathbb{E}[C(\widetilde{\pattern})]\le C(\alg)$.
For each $X\in\mathcal{F}\setminus\{X_0\}$, the definition of $\delta$-obliviousness implies \(\Pr[\widetilde{\pattern}\notin\operatorname{supp}(\pattern_{\alg}(X))]\le \delta\). By a union bound, the event \(\mathcal{E}:=\{\widetilde{\pattern}\in \bigcap_{X\in\mathcal{F}} \operatorname{supp}(\pattern_{\alg}(X))\}\) has probability at least $1-(|\mathcal{F}|-1)\delta$. Hence \(\mathbb{E}[C(\widetilde{\pattern})\mid \mathcal{E}]\le\mathbb{E}[C(\widetilde{\pattern})]/\Pr[\mathcal{E}]\le C(\alg)/(1-(|\mathcal{F}|-1)\delta)\).
Therefore there exists a communication pattern $\pattern^\star\in \bigcap_{X\in\mathcal{F}} \operatorname{supp}(\pattern_{\alg}(X))$ whose cost is at most $\frac{C(\alg)}{1-(|\mathcal{F}|-1)\delta}$. Since $\alg$ is always correct, for every $X\in\mathcal{F}$ there is an execution that produces a valid output while following $\pattern^\star$. The ideal-model algorithm $\alg^\star$ can deterministically select this execution for input $X$.
\end{proof}

By the two reductions above, to prove a lower bound it suffices to fix a finite family $\mathcal F$ of inputs with the same leakage and analyze a deterministic perfectly communication-oblivious algorithm in the ideal model. The reduction incurs only an $o(n)$ loss when $(|\mathcal F|-1)(\delta+2\delta')=o(1)$ and $\delta'C(\alg_{\mathrm{fb}})=o(n)$; in particular, these conditions hold when $|\mathcal F|$ is polynomial in $n$ and both the failure and non-obliviousness probabilities are negligible.

For such a deterministic algorithm, perfect communication obliviousness gives a common pattern $\pattern=(L^1,\ldots,L^R)$ for every input in $\mathcal F$. We use a static fractional relaxation: give each item mass $p/n$, allow this mass to be split, and discard the timing of communication. Aggregate the pattern into the directed capacity graph $w$, where $w_{ij}:=(p/n)\sum_{r=1}^R L_{ij}^r$, and define its normalized cost by $C(w):=p^{-1}\sum_{i,j}w_{ij}=C(\pattern)/n$.

For every $X\in\mathcal F$, relaxing a correct execution yields a fractional routing from the normalized initial placement to a valid output, with flow at most $w_{ij}$ on every arc $i\to j$. We say that $w$ supports $X$ when such a routing exists. Thus the fixed graph $w$ supports every input in $\mathcal F$, although the routing and, when the output is not unique, the terminal assignment may depend on $X$. Consequently, the minimum of $C(w)$ over all capacity graphs supporting $\mathcal F$ is a lower bound on the normalized communication cost $C(\pattern)/n$. This minimum-capacity problem is a directed variant of robust network design with dynamic routing~\cite{goyal2011dynamic}. Since divisibility and the removal of temporal constraints only relax the original model, multiplying this minimum by $n$ gives a communication-cost lower bound for the original indivisible BSP problem. Accordingly, all lower bounds below are established in this relaxed capacity model.

\section{Permutation}

\subsection{Lower bound}
For the lower bound, it suffices to consider balanced permutation. For $i,j\in[p]$, let $X_{ij}:=\{x\in X_i:t_x=j\}$. After normalization, a balanced input $X$ is represented by the demand matrix $\bigl((p/n)|X_{ij}|\bigr)_{i,j}$. Its $i$-th row sums to $(p/n)|X_i|=1$, while its $j$-th column sums to $(p/n)d_j=1$; hence it is doubly stochastic. The direct-versus-detour argument of Keslassy et al.~\cite{keslassy2005optimal} yields the following bound; Appendix~\ref{app:perm-normalized-lower-bound} gives the adaptation.

\begin{theorem}[Permutation lower bound]\label{lb_cyclic_shift}
Any normalized capacity matrix \(w\) supporting every divisible balanced-permutation demand satisfies $p^{-1}\sum_{i,j}w_{ij}\ge2-2/p$. Consequently, even in the ideal model, every deterministic perfectly communication-oblivious algorithm for balanced permutation on $n$ indivisible items has communication cost at least $(2-2/p)n$.
\end{theorem}

Combining the normalized deterministic bound with failure removal and finite-family derandomization gives the following extension to randomized algorithms.

\begin{corollary}
\label{cor:permlower}
    A randomized permutation algorithm whose failure probability and non-oblivious probability are negligible has communication cost at least $(2-2/p)n-o(n)$.
\end{corollary}
\begin{proof}
    The always-correct non-oblivious permutation algorithm has communication cost $n$, and the family $\cycf$ defined in Appendix~\ref{app:perm-normalized-lower-bound} has size $p-1$.  Plugging these facts into Lemma \ref{lem:failure_removal} and Lemma \ref{lem:finite_family_derand} yields the conclusion.
\end{proof}

\subsection{Upper bound}
\label{sec:permUB}
The following two-round construction is a deterministic implementation of Valiant load balancing~\cite{valiant1981universal}.

\paragraph{Spray and count dissemination.}
Processor $i$ locally sorts its items by destination, using a fixed tie-breaking rule, and assigns them cyclically to the $p$ intermediate processors. For every $u\ne i$, processor $i$ sends the items assigned to $u$ in a payload padded to $\lceil n/p^2\rceil$ positions and appends its local destination-count vector $(|X_{i1}|,\ldots,|X_{ip}|)$. Thus every processor reconstructs $D$ and receives $(p-1)(\lceil n/p^2\rceil+p)$ words in this round.

\paragraph{Forward.}
Because the items in $X_{ij}$ are consecutive in this local order, they contribute at most $\lfloor |X_{ij}|/p\rfloor+1$ items to any intermediate processor. Hence every intermediate processor holds at most $\sum_i(\lfloor |X_{ij}|/p\rfloor+1)\le\lfloor d_j/p\rfloor+p$ items destined for processor $j$. Each intermediate processor $u$ sends its items for every $j\ne u$ in a payload padded to this bound. Every item therefore reaches its prescribed destination, and processor $j$ terminates with exactly $d_j$ items.

The first-round message lengths are $\lceil n/p^2\rceil+p$, while the message from $u$ to $j$ in the second round has length $\lfloor d_j/p\rfloor+p$. These lengths depend only on $n,p$, and the authorized leakage $D$. Together with the fixed local sorting and cyclic assignment, this makes the algorithm deterministic and perfectly communication-oblivious. Its communication cost satisfies $C(\pattern)=p(p-1)(\lceil n/p^2\rceil+p)+(p-1)\sum_j(\lfloor d_j/p\rfloor+p)\le(2-2/p)n+O(p^3)$. For an input with destination counts $D$, write $d_{\max}:=\max_jd_j$. Each processor receives at most $(p-1)(\lfloor d_{\max}/p\rfloor+p)$ words in the forward round. Together with the spray-round bound, this gives at most $(1-1/p)d_{\max}+O(p^2)$ received words per processor per round.

\begin{theorem}[Permutation upper bound]\label{thm:deterministic_oblivious}
There is a deterministic perfectly communication-oblivious algorithm for the permutation problem in the realistic model. It uses two rounds and has communication cost at most $(2-\frac{2}{p})n+O(p^3)$. On every input with destination-count vector $D$, each processor receives at most $(1-1/p)\max_j d_j+O(p^2)$ words in any round.
\end{theorem}
When $p=o(n^{1/3})$, the communication overhead is $o(n)$ and the additive per-round receive overhead is $O(p^2)=o(n/p)$. In particular, balanced permutation has load at most $n/p+o(n/p)$.

\section{Sorting}

As argued earlier, sorting and balanced permutation are equivalent in the ideal model, so the $(2-2/p)n$ lower bound in Theorem~\ref{lb_cyclic_shift}, as well as the randomized lower bound in Corollary \ref{cor:permlower}, also holds for sorting. %
We therefore focus on the upper bound. We reduce sorting to a nearly balanced permutation problem and solve it using the permutation algorithm from Section~\ref{sec:permUB}. The basic idea follows classical sample sort~\cite{regularSampleSort}, with additional care to ensure obliviousness and match the lower bound up to lower-order terms.

\subsection{Algorithm}
Fix parameters $s,\tau,\epsilon$, whose values will be determined later. Let \(\alg_{\mathrm{sort}}\) denote the following three-phase algorithm.

\paragraph{Sampling.}
Processor $1$ first draws a sample (without replacement) of size $ps$ obliviously, i.e., the communication pattern is fixed regardless of which items are sampled. To do so, each processor uniformly samples $\tau$ ($>s$) key copies from its local data, retains all original items, and sends the copies to processor~1. Processor~1 then generates a random vector $(s_1,\ldots,s_p)$ from the \emph{multivariate hypergeometric} distribution with parameters $n$, $(n/p,\ldots,n/p)$, and $ps$. Here $s_i$ is the number of items taken from processor~$i$ in a uniform random sample of size $ps$. If $s_i>\tau$ for any $i$, processor~1 announces failure; otherwise, it uniformly draws $s_i$ of the $\tau$ key copies originating from processor~$i$.

Next, processor $1$ sorts the sample and lets \(q_k\) be the key of sample rank \(sk\) for \(k\in[p-1]\). It then broadcasts \((q_1,\ldots,q_{p-1})\).

\paragraph{Permutation.}
Set \(q_0:=-\infty\) and \(q_p:=\infty\). For every \(i,j\in[p]\), let \(X_{ij}:=\{x\in X_i:q_{j-1}<t_x\le q_j\}\). The processors then invoke the permutation procedure from Section~\ref{sec:permUB} on this partition, with \(j\) as the prescribed destination of every item in \(X_{ij}\). After the spray-and-count dissemination round, every processor knows \(D=(d_1,\ldots,d_p)\) and checks whether \(d_j>(1+\epsilon)n/p\) for any \(j\in[p]\). If so, they abort; otherwise, they complete the forward round.

\paragraph{Rebalancing.}
After the permutation step, processor~$i$ holds the items in the $i$-th pivot interval, but these items need not yet be locally ordered. Each processor first sorts its items locally. The distribution may be slightly out of balance, with processor~$i$ holding $d_i\le (1+\epsilon)n/p$ items. In the last round, each processor sends the items whose ranks lie outside $((i-1)n/p,in/p]$ to their correct destinations to form the balanced sorting output.

\subsection{Analysis}
\begin{theorem}[Sorting upper bound]\label{thm:sorting_upper_bound}
For any $\epsilon,\delta\in(0,1)$, define
\[
s := \left\lceil
\frac{8(1+\epsilon/2)}{\epsilon^2}\ln\frac{4p}{\epsilon\delta}
\right\rceil,
\qquad
\tau := \left\lceil
s + \sqrt{3s\ln\frac{2p}{\delta}}
\right\rceil,
\]
and suppose that $\tau\le n/p$. Then the sorting algorithm above uses five rounds, is $\delta$-oblivious, correctly solves the sorting problem except with probability at most $\delta$, has load at most $\max\{p\tau,(1+\epsilon)n/p\}+O(p^2)$, and incurs expected communication cost at most
\[
\Bigl(2-\frac{2}{p}\Bigr)n
+O\Bigl(\frac{p}{\epsilon^2}\log\frac{p}{\epsilon\delta}
+p^3+n\sqrt{p/s}\Bigr).
\]
\end{theorem}
For $p^3=o(n)$, choose $\epsilon=\sqrt{p/n^{1/3}}$ and $\delta=\exp(-n^{1/3})$. Then $\epsilon=o(1)$, $s=\Theta(n^{2/3}/p)$, and $\tau=(1+o(1))s\le n/p$ for all sufficiently large $n$. Since $p\tau=\Theta(n^{2/3})=o(n/p)$ and $n\sqrt{p/s}=O(pn^{2/3})=o(n)$, the theorem gives expected communication $(2-2/p)n+o(n)$ and load at most $n/p+o(n/p)$, with negligible failure probability and obliviousness error.

The rest of this section proves Theorem~\ref{thm:sorting_upper_bound}.

The sampling procedure is the algorithm's only randomized step. To analyze it, imagine drawing an unrestricted uniform sample of $ps$ items, and let $\widehat{s}_i$ be its number of items from processor~$i$. The vector $(\widehat{s}_i)_{i\in[p]}$ has the same distribution as the algorithm's $(s_i)_{i\in[p]}$, and, conditional on it, the local samples are independent uniform $\widehat{s}_i$-subsets. If every $\widehat{s}_i\le\tau$, each local sample can be extended uniformly to a local $\tau$-subset; by symmetry, reversing these two choices gives exactly the implemented sampler. Thus the two samples can be coupled to agree outside $\mathcal{E}_{\mathrm{samp}}:=\{\exists i:\widehat{s}_i>\tau\}$. Let \(\mathcal{E}_{\mathrm{bucket}}\) be the event that some pivot bucket of the unrestricted sample contains more than \((1+\epsilon)n/p\) items.

\begin{lemma}[Failure probability]\label{lem:sorting_failure}
Under the parameter choice of Theorem~\ref{thm:sorting_upper_bound}, $\Pr(\mathcal{E}_{\mathrm{samp}})\le\delta/2$ and $\Pr(\mathcal{E}_{\mathrm{bucket}})\le\delta/2$. Consequently, the sorting algorithm fails with probability at most~$\delta$.
\end{lemma}

\begin{proof}
Each $\widehat{s}_i$ is hypergeometric with mean $s$. Since $s\ge12\ln(2p/\delta)$, the additive Chernoff bound at deviation $\sqrt{3s\ln(2p/\delta)}\le s$, followed by a union bound, gives $\Pr(\mathcal{E}_{\mathrm{samp}})\le p\cdot\delta/(2p)=\delta/2$.

Now consider rank intervals of length $\lceil(1+\epsilon/2)n/p\rceil$ whose starting ranks are $1+a\lceil\epsilon n/(2p)\rceil$ and that lie in $[n]$, where $a$ ranges over nonnegative integers. There are at most $2p/\epsilon$ of them. Every interval longer than $(1+\epsilon)n/p$ contains one: its first grid point is fewer than $\lceil\epsilon n/(2p)\rceil$ ranks from its left endpoint, while the length slack is at least $\lceil\epsilon n/(2p)\rceil-1$. Each pivot bucket contains exactly $s$ sampled items, so $\mathcal{E}_{\mathrm{bucket}}$ implies that some test interval contains at most $s$ of them. For any fixed test interval, the sample count is hypergeometric with mean at least $s(1+\epsilon/2)$, so
\[
\Pr(\text{the interval contains at most $s$ sampled items})
\le \exp\!\left(-\frac{s\epsilon^2}{8(1+\epsilon/2)}\right).
\]
The definition of $s$ makes this at most $\epsilon\delta/(4p)$; another union bound gives $\Pr(\mathcal{E}_{\mathrm{bucket}})\le\delta/2$. On $\mathcal{E}_{\mathrm{samp}}^c$ the two samples coincide, so the actual algorithm can fail only in $\mathcal{E}_{\mathrm{samp}}\cup\mathcal{E}_{\mathrm{bucket}}$.
\end{proof}

\begin{lemma}[Obliviousness]\label{lem:sorting_oblivious}
Under the parameter choice of Theorem~\ref{thm:sorting_upper_bound}, the sorting algorithm is $\delta$-oblivious with respect to $\Phi(X):=n$.
\end{lemma}

\begin{proof}
The sampling pattern is fixed: every noncollector sends exactly $\tau$ key copies, and processor~1 broadcasts either the $p-1$ pivots or a failure flag padded to the same length. Let $\pattern_{\mathrm{ref}}$ be the pattern of the execution that uses the unrestricted sample and aborts on bucket overflow. Its distribution is input-independent: sampled ranks are uniform, the resulting bucket sizes determine $D$, the permutation subroutine's pattern is determined by $D$, and the bucket sizes also determine the overflow and rebalancing messages. Under the coupling above, the actual and reference patterns agree on $\mathcal{E}_{\mathrm{samp}}^c$. Hence Lemma~\ref{lem:sorting_failure} gives $\Delta(\pattern_{\alg_{\mathrm{sort}}}(X),\pattern_{\mathrm{ref}})\le\delta/2$ for every input $X$; for $X,X'$ the triangle inequality gives
\[
\Delta\bigl(\pattern_{\alg_{\mathrm{sort}}}(X),\pattern_{\alg_{\mathrm{sort}}}(X')\bigr)
\le \delta/2+\delta/2=\delta.\qedhere
\]
\end{proof}

\begin{lemma}[Resource bounds]\label{lem:sorting_resources}
The sorting algorithm uses five rounds and satisfies the load and expected communication bounds in Theorem~\ref{thm:sorting_upper_bound}.
\end{lemma}

\begin{proof}
Sample collection and pivot broadcast use two rounds, permutation uses two, and rebalancing uses one.

Sampling communicates at most $p\tau+O(p^2)$ words and has load at most $p\tau$. By Theorem~\ref{thm:deterministic_oblivious}, permutation communicates at most $(2-2/p)n+O(p^3)$ words. The spray round has load at most $n/p+O(p^2)$, and the bucket-size check bounds the forward-round load by $(1+\epsilon)n/p+O(p^2)$. Rebalancing has load at most $n/p$, giving an overall load bound of $\max\{p\tau,(1+\epsilon)n/p\}+O(p^2)$. This accounts for all communication except rebalancing.

Let $Q_k$ be the global rank of the reference pivot $q_k$, namely the $(ks)$-th unrestricted-sample order statistic. Exactly $|Q_k-kn/p|$ items cross logical boundary~$k$, so $\sum_{k=1}^{p-1}|Q_k-kn/p|$ upper-bounds direct transmissions, since one transmission may cross several boundaries. The actual algorithm either uses these pivots or aborts with zero rebalancing cost, so this domination is unconditional.

Writing the unsampled ranks as gaps around the $ps$ sampled ranks gives $\mathbb E[Q_k]=ks(n+1)/(ps+1)$ and
\[
\operatorname{Var}(Q_k)
=\frac{sk\bigl(s(p-k)+1\bigr)(n+1)(n-ps)}{(ps+1)^2(ps+2)}.
\]
Thus $\sqrt{\operatorname{Var}(Q_k)}\le n\sqrt{k}/(p\sqrt{s})$ and $\mathbb E[Q_k]-kn/p=-k(n-ps)/(p(ps+1))$. Using $\mathbb E|Q_k-kn/p|\le\sqrt{\operatorname{Var}(Q_k)}+|\mathbb E[Q_k]-kn/p|$, $\sum_{k=1}^{p-1}\sqrt{k}<p\sqrt p$, and
$\sum_{k=1}^{p-1}k(n-ps)/(p(ps+1))=(p-1)(n-ps)/(2(ps+1))\le n/(2s)$, we obtain
\[
\sum_{k=1}^{p-1}\mathbb E\left|Q_k-\frac{kn}{p}\right|
\le n\sqrt{p/s}+\frac{n}{2s}.
\]
Uniformity of the sampled global ranks makes this expectation uniform over inputs.

Since $\tau=O(s)=O(\bigl(1/\epsilon^2\bigr)\log(p/(\epsilon\delta)))$ and $1/s\le\sqrt{p/s}$, adding the sampling, permutation, and expected rebalancing costs proves the communication bound.
\end{proof}

If the algorithm does not abort, Theorem~\ref{thm:deterministic_oblivious} guarantees that permutation gathers the ordered pivot intervals, local sorting orders them, and rebalancing places every rank on its prescribed processor. Together with Lemmas~\ref{lem:sorting_failure}, \ref{lem:sorting_oblivious}, and~\ref{lem:sorting_resources}, this proves Theorem~\ref{thm:sorting_upper_bound}.

\section{0-1 Sorting}
\label{sec:01-sorting}

The cyclic-shift hard family underlying Theorem~\ref{lb_cyclic_shift} can be encoded as stable 0-1-sorting instances. Hence stable 0-1 sorting inherits the $(2-2/p)n$ lower bound for general sorting. We therefore focus on unstable 0-1 sorting, where equal-key items may exchange valid destinations; this flexibility is what permits a smaller communication coefficient.

Since equal-key items are interchangeable, we describe an input by the amount of zero mass on each processor. For an input \(X\), let \(z:=z_X\in[0,1]^p\) be its normalized zero-count vector, where \(z_i:=(p/n)\lvert\{x\in X_i:t_x=0\}\rvert\). Processor \(i\) initially holds zero mass \(z_i\) and one mass \(1-z_i\), and the total zero mass is \(\lVert z\rVert_1\).

In the sorted output, zero mass fills the processors from left to right. Its distribution is therefore the \emph{sorted zero profile} \(\rho(z):=(\underbrace{1,\ldots,1}_{\ell},\beta,0,\ldots,0)\), where \(\ell:=\lfloor\lVert z\rVert_1\rfloor\) and \(\beta:=\lVert z\rVert_1-\ell\in[0,1)\); the \(\beta\) entry is omitted when \(\ell=p\).

We call an input \emph{pure} if every processor initially holds either all zeros or all ones, and write \(\mathcal F_{\mathrm{pure}}:=\{0,1\}^p\) for these input vectors. We use pure inputs to establish the lower bound and later extend their routing schemes to arbitrary inputs.

\subsection{The Cut-LP and the lower bound}
\label{sec:01-lp-characterization}

Consider any set of processors. Its zero surplus must leave the set, while a zero deficit requires an equal amount of one mass to leave, since its total mass is unchanged by sorting. The common capacity graph must accommodate the largest such demand over all pure inputs. These necessary cut constraints yield the following LP lower bound.

Let \(w=(w_{ij})\) be the common normalized capacity matrix, with cost \(C(w):=p^{-1}\sum_{i,j}w_{ij}\). For \(U,V\subseteq[p]\), write \(w(U,V):=\sum_{i\in U,j\in V}w_{ij}\), and let \(\overline U:=[p]\setminus U\).

For an input vector \(z\), define the signed zero surplus at processor \(i\) by \(b_i(z):=z_i-\rho_i(z)\), and its total over \(U\) by \(b_z(U):=\sum_{i\in U}b_i(z)\). The largest cut demand over all pure inputs is \(\operatorname{dem}(U):=\max_{z\in\mathcal F_{\mathrm{pure}}}\lvert b_z(U)\rvert\). Since total zero mass is preserved, \(b_z(\overline U)=-b_z(U)\); hence \(\operatorname{dem}(U)=\operatorname{dem}(\overline U)\), while \(\operatorname{dem}(\emptyset)=\operatorname{dem}([p])=0\). Consider the following Cut-LP:
\begin{equation*}
\begin{array}{ll}
\text{minimize} & \displaystyle C(w) \\[2mm]
\text{subject to}
& \displaystyle w(U,\overline U)\ge \operatorname{dem}(U)
\qquad(\emptyset\ne U\subsetneq[p]),\\[2mm]
& w_{ij}\ge0\qquad(i,j\in[p]).
\end{array}
\tag{\ensuremath{\mathrm{Cut}(p)}}\label{lp:01-cut}
\end{equation*}
Let \(C_{\mathrm{cut}}(p)\) denote the optimal value of~\eqref{lp:01-cut}.

\begin{theorem}[0-1 sorting lower bound]
\label{lem:01-cut-necessity}
Any normalized capacity matrix \(w\) supporting every pure 0-1 sorting input in the divisible setting satisfies $C(w)\ge C_{\mathrm{cut}}(p)$. Consequently, even in the ideal model, deterministic perfectly communication-oblivious 0-1 sorting on \(n\) indivisible items requires communication at least \(nC_{\mathrm{cut}}(p)\).
\end{theorem}

\begin{proof}
For the first statement, let \(w\) be such a shared aggregate capacity matrix. Fix \(z\in\mathcal F_{\mathrm{pure}}\) and a nontrivial cut \(U\). If \(b_z(U)\ge0\), then zero mass of amount \(b_z(U)\) must leave \(U\). If \(b_z(U)<0\), then one mass of amount \(-b_z(U)\) must leave \(U\), because every processor has the same total mass initially and finally. Thus \(w(U,\overline U)\ge\lvert b_z(U)\rvert\) for every pure input; maximizing over \(z\) gives \(w(U,\overline U)\ge\operatorname{dem}(U)\). Hence \(w\) is feasible for~\eqref{lp:01-cut}, whose objective is $C(w)$.
\end{proof}

\begin{corollary}
\label{cor:01lower}
    Any randomized 0-1 sorting algorithm whose failure probability and non-oblivious probability are negligible must have communication cost at least $C_{\mathrm{cut}}(p)n-o(n)$.
\end{corollary}
\begin{proof}
By Lemma~\ref{lem:01-sparse-hard-family}, the proof of Theorem~\ref{lem:01-cut-necessity} yields the same deterministic lower bound \(nC_{\mathrm{cut}}(p)\) on a family of at most \(p^2\) inputs. The result then follows from Lemmas~\ref{lem:failure_removal} and~\ref{lem:finite_family_derand}.
\end{proof}

\subsection{The Route-LP and the upper bound}
\label{sec:01-route-characterization}

\subsubsection{The finite-round Route-LP}
\label{sec:upper_lp_formulation}

Fix \(R\ge1\). The LP \(\mathrm{Route}(p,R)\) describes a public \(R\)-round communication pattern that simultaneously supports every pure input. Its normalized public lengths are shared by all inputs; for each pure input, separate flow and state variables certify a correct routing under those lengths.

For \(r\in[R]\) and \(i\ne j\), let \(\mathcal L_{ij}^r\) be the normalized public message length from processor \(i\) to processor \(j\) in round \(r\). For each \(z\in\mathcal F_{\mathrm{pure}}\) and \(\sigma\in\{0,1\}\), let \(f_{ij}^{\sigma,z,r}\) be the amount of key class \(\sigma\) sent on that arc, and let \(m_i^{\sigma,z,r}\) be the amount held by processor \(i\) at the beginning of round \(r\), with \(r=R+1\) denoting the terminal state.
All length, flow, and state variables are nonnegative. Each constraint below is imposed for every choice of its free indices, where \(z\in\mathcal F_{\mathrm{pure}}\), \(\sigma\in\{0,1\}\), \(r\in[R]\), \(i\in[p]\), and \(j\in[p]\setminus\{i\}\). The objective and remaining constraints are
\begin{equation*}
\begin{aligned}
\text{minimize}\quad&
\textstyle\frac1p\sum_{r=1}^{R}\sum_{i\ne j}\mathcal L_{ij}^r
&&\\[1ex]
\text{subject to}\quad&
m_i^{\sigma,z,r+1}-m_i^{\sigma,z,r}
-\textstyle\sum_{j\ne i}f_{ji}^{\sigma,z,r}
+\textstyle\sum_{j\ne i}f_{ij}^{\sigma,z,r}=0
&&\text{(Flow)}\\
&
m_i^{0,z,1}=z_i,\qquad
m_i^{1,z,1}=1-z_i
&&\text{(Init)}\\
&
m_i^{0,z,R+1}=\rho_i(z),\qquad
m_i^{1,z,R+1}=1-\rho_i(z)
&&\text{(Term)}\\
&
\textstyle\sum_{j\ne i}f_{ij}^{\sigma,z,r}\le m_i^{\sigma,z,r}
&&\text{(Avail)}\\
&
f_{ij}^{0,z,r}+f_{ij}^{1,z,r}\le \mathcal L_{ij}^{r}
&&\text{(Len)}.
\end{aligned}
\tag{\ensuremath{\mathrm{Route}(p,R)}}\label{lp:01-cr}
\end{equation*}
Constraints (Flow), (Init), and (Term) impose conservation and the prescribed initial and terminal states. Constraint (Avail) enforces local availability, whereas (Len) requires the two key classes to share the public message length. Unused message space may be filled with padding. For a feasible solution \(\mathcal S\), write \(C(\mathcal S)\) for its objective value. The optimal value of~\eqref{lp:01-cr} is \(C_{\mathrm{route}}(p,R)\).
The LP has size \(O(2^pRp^2)\), because it contains one state-and-flow block for each pure input.

We measure receive load in units of $n/p$. For a feasible solution $\mathcal S$, define its normalized load by
\begin{equation*}
\overline{\operatorname{load}}(\mathcal S)
:=\max_{r\in[R],\,j\in[p]}\sum_{i\ne j}\mathcal L_{ij}^r.
\end{equation*}

\begin{lemma}[Extension from pure inputs]\label{lem:lp_superposition}
Every feasible solution of \(\mathrm{Route}(p,R)\) extends to a correct divisible routing scheme for every input vector \(z\in[0,1]^p\) under the same communication pattern. The extension preserves normalized load.
\end{lemma}

\begin{proof}
Fix \(z\in[0,1]^p\), write \(\lVert z\rVert_1=\ell+\beta\), and let \(\mathcal F_k:=\{z'\in\mathcal F_{\mathrm{pure}}:\lVert z'\rVert_1=k\}\), where \(\mathcal F_k:=\emptyset\) for \(k\notin\{0,\ldots,p\}\). The standard hypersimplex decomposition~\cite{ziegler_polytopes} gives coefficients \(\lambda_{z'}\ge0\) such that
\[
z=\sum_{z'\in\mathcal F_\ell}\lambda_{z'}z'
  +\sum_{z'\in\mathcal F_{\ell+1}}\lambda_{z'}z',
\qquad
\sum_{z'\in\mathcal F_\ell\cup\mathcal F_{\ell+1}}\lambda_{z'}=1,
\qquad
\sum_{z'\in\mathcal F_{\ell+1}}\lambda_{z'}=\beta.
\]
Choose the lexicographically first basic feasible coefficient vector; it has support at most \(p+1\). Superpose the corresponding pure-input routing schemes:
\[
f_{ij}^{\sigma,z,r}:=\sum_{z'\in\mathcal F_\ell\cup\mathcal F_{\ell+1}}\lambda_{z'} f_{ij}^{\sigma,z',r},
\qquad
m_i^{\sigma,z,r}:=\sum_{z'\in\mathcal F_\ell\cup\mathcal F_{\ell+1}}\lambda_{z'} m_i^{\sigma,z',r}.
\]
Linearity preserves (Flow), (Avail), and (Len), and gives the initial state \((z_i,1-z_i)\). Moreover, for every \(i\in[p]\), \(m_i^{0,z,R+1}=\sum_{z'\in\mathcal F_\ell\cup\mathcal F_{\ell+1}}\lambda_{z'}\rho_i(z')=\rho_i(z)\), because the coefficients from \(\mathcal F_{\ell+1}\) have total weight \(\beta\). The same calculation gives the terminal one-mass vector. The extension uses the same communication pattern and hence preserves normalized load. Thus the superposed scheme is correct and deterministic under the fixed choice of coefficients.
\end{proof}

Lemma~\ref{lem:lp_superposition} routes every divisible input under the shared Route-LP lengths, giving the following forward construction.

\begin{corollary}[Route-LP forward construction]
\label{cor:01-route-to-divisible-algorithm}
Every feasible solution \(\mathcal S\) of \(\mathrm{Route}(p,R)\) induces a deterministic perfectly communication-oblivious \(R\)-round algorithm for divisible 0-1 sorting in the ideal model. Its normalized communication is $C(\mathcal S)$, and its normalized load is $\overline{\operatorname{load}}(\mathcal S)$.
\end{corollary}

\subsubsection{\texorpdfstring{Finite-round bounds for small \(p\)}{Finite-round bounds for small p}}
\suppressfloats[t]

We computed the Cut-LP and Route-LP for \(2\le p\le11\) to compare the static lower bound with achievable finite-round communication costs. Table~\ref{tab:01-results} reports these bounds alongside the general-sorting coefficient \(2-2/p\), together with the round count and normalized load of each route. The Cut-LP optima are certified by matching rational primal and dual solutions, and the Route-LP upper bounds by rational feasible solutions. All certificates are verified using exact arithmetic.

\begin{table}[!ht]
\centering
\small
\setlength{\tabcolsep}{10pt}
\renewcommand{\arraystretch}{1.15}
\begin{tabular}{@{}cccccc@{}}
\toprule
\multicolumn{3}{c}{} & \multicolumn{3}{c}{Route-LP solution \(\mathcal S\)} \\
\cmidrule(l){4-6}
\(p\) & \(2-2/p\) & \(C_{\mathrm{cut}}(p)\) & \(C(\mathcal S)\) & \(R\) & \(\overline{\operatorname{load}}(\mathcal S)\) \\
\midrule
2 & 1 & 1 & \(\boldsymbol{1}\) & 1 & \(1\) \\
3 & \(4/3\) & 1 & \(\boldsymbol{1}\) & 3 & \(1\) \\
4 & \(3/2\) & 1 & 1.0011 & 400 & \(3/665\) \\
5 & \(8/5\) & \(6/5\) & \(\boldsymbol{6/5}\) & 5 & \(4/3\) \\
6 & \(5/3\) & \(11/9\) & \(\boldsymbol{11/9}\) & 5 & \(4/3\) \\
7 & \(12/7\) & \(9/7\) & \(\boldsymbol{9/7}\) & 6 & \(12/11\) \\
8 & \(7/4\) & \(37/28\) & \(\boldsymbol{37/28}\) & 5 & \(1.2659\) \\
9 & \(16/9\) & \(451/342\) & \(\boldsymbol{451/342}\) & 6 & \(1.0385\) \\
10 & \(1.8000\) & \(1.3375\) & \(1.3441\) & 4 & \(1.1688\) \\
11 & \(1.8182\) & \(1.3438\) & \(1.4062\) & 3 & \(1.5221\) \\
\bottomrule
\end{tabular}
\caption{Cut-LP lower bounds and finite-round Route-LP upper bounds for small \(p\). Bold route costs match the Cut-LP lower bound exactly. Cut-LP optima are displayed as fractions or rounded decimals; decimal route costs and loads are rounded upward.}
\label{tab:01-results}
\end{table}

For every \(3\le p\le11\), the displayed route cost is strictly below \(2-2/p\), showing that the communication advantage of 0-1 sorting can already be realized by finite-round algorithms. Moreover, for \(p=2,3,5,\ldots,9\), the route cost matches \(C_{\mathrm{cut}}(p)\) exactly, attaining the Cut-LP lower bound in only one to six routing rounds. Our computations identify the reported \(R\) as the minimum number of rounds needed to attain this optimum.

For \(p=4\), the 400-round route has cost at most \(1.0011\), within \(0.11\%\) of \(C_{\mathrm{cut}}(4)=1\), but none of the computed routes attains the lower bound. We conjecture that no finite-round Route-LP solution attains \(C_{\mathrm{cut}}(4)\). The unmatched bounds for \(p=10,11\) use only four and three rounds, respectively; larger horizons were not computed.

This raises the general question of whether increasing the number of rounds always closes the gap to the Cut-LP bound: for every \(p\ge2\), do we have
\[
\inf_{R\ge1}C_{\mathrm{route}}(p,R)
\stackrel{?}{=}
C_{\mathrm{cut}}(p)?
\]

\subsubsection{From cuts to finite-round routes}
\label{sec:01-cut-to-route}

We answer this question affirmatively by converting any feasible Cut-LP solution into a finite-round Route-LP solution with vanishing communication-cost overhead.

\begin{lemma}[Cut-to-route conversion]
\label{lem:01-cut-to-route}
Let $w$ be a feasible solution of $\mathrm{Cut}(p)$. For every integer $R\ge p+2$, there is a feasible solution of $\mathrm{Route}(p,R)$ with objective value at most $C(w)+2(p-1)^2/(p(R-2))$ and normalized load at most $pC(w)/(R-2)$.
\end{lemma}

\begin{proof}
Put \(T:=R-2\ge p\), and replace each pair \(w_{ij},w_{ji}\) by its average. This preserves the objective. For every nontrivial \(U\), the new outgoing capacity is \((w(U,\overline U)+w(\overline U,U))/2\ge(\operatorname{dem}(U)+\operatorname{dem}(\overline U))/2=\operatorname{dem}(U)\), so feasibility is preserved. We may therefore assume that \(w\) is symmetric.

Fix \(z\in\mathcal F_{\mathrm{pure}}\). Feasibility of \(w\) gives \(\lvert b_z(U)\rvert\le w(U,\overline U)\) for every \(U\subseteq[p]\).

\paragraph{Paired bidirectional key flows.}
Since \(\sum_i b_i(z)=0\), max-flow/min-cut applied to the cut inequalities yields a nonnegative zero-flow \(f^0\) satisfying \(f^0_{ij}\le w_{ij}\) and \(\sum_{j\ne i}(f^0_{ij}-f^0_{ji})=b_i(z)\) for every \(i\in[p]\). Cancel antiparallel flow and directed cycles, and define the one-flow by \(f^1_{ij}:=f^0_{ji}\). It has divergence \(-b_i(z)\) at each processor \(i\). Moreover, antiparallel cancellation and symmetry give \(f^0_{ij}+f^1_{ij}=f^0_{ij}+f^0_{ji}\le w_{ij}\) on every directed arc \(i\to j\), so the capacities suffice to carry both key flows simultaneously. Decompose \(f^0\) into simple surplus-to-deficit paths \(P\) with amounts \(f_P\); their reversals decompose \(f^1\). Routing zero mass \(f_P\) along each \(P\) and one mass \(f_P\) along its reverse therefore produces the prescribed terminal states.

\begin{figure}[!ht]
\centering
\begin{minipage}[t]{0.36\textwidth}
\centering
\vspace{0pt}
\textbf{(a) Paired bidirectional flows}\par\smallskip
\begin{tikzpicture}[
  x=1cm,y=1cm,font=\footnotesize,>=stealth,
  vertex/.style={circle,draw=black!70,fill=white,minimum size=0.62cm,inner sep=0pt},
  zero/.style={draw=blue!65!black,semithick,->},
  one/.style={draw=orange!80!black,semithick,densely dashed,->}
]
  \node[vertex] (v0) at (0.35,1.55) {$v_0$};
  \node[vertex] (v1) at (1.75,1.55) {$v_1$};
  \node[vertex] (v2) at (3.15,1.55) {$v_2$};
  \node[vertex] (v3) at (4.55,1.55) {$v_3$};

  \draw[zero,bend left=13] (v0) to (v1);
  \draw[zero,bend left=13] (v1) to (v2);
  \draw[zero,bend left=13] (v2) to (v3);
  \draw[one,bend left=13] (v3) to (v2);
  \draw[one,bend left=13] (v2) to (v1);
  \draw[one,bend left=13] (v1) to (v0);

  \node[blue!65!black] at (2.45,2.02) {zero $f_P$};
  \node[orange!80!black] at (2.45,1.08) {one $f_P$};
  \node[align=center] at (0.35,0.08) {$b_{v_0}(z)>0$};
  \node[align=center] at (4.55,0.08) {$b_{v_3}(z)<0$};
\end{tikzpicture}
\end{minipage}\hfill
\begin{minipage}[t]{0.61\textwidth}
\centering
\vspace{0pt}
\textbf{(b) Pipeline and cleanup}\par\smallskip
\begin{tikzpicture}[x=0.80cm,y=0.68cm,font=\scriptsize,>=stealth]
  \node at (4.35,4.15) {$p=4$, $T=7$; each cohort has mass $f_P/T$};

  \node[anchor=east] at (0.68,3.83) {$r$};
  \foreach \r in {1,...,7}
    \node at ({0.35+\r},3.83) {$\r$};
  \node[anchor=east] at (0.68,3.1) {$e_1$};
  \node[anchor=east] at (0.68,2.3) {$e_2$};
  \node[anchor=east] at (0.68,1.5) {$e_3$};

  \foreach \r/\c in {1/1,2/2,3/3,4/4} {
    \fill[blue!12] ({0.85+\r-1},2.7) rectangle ({0.85+\r},3.5);
    \node at ({0.35+\r},3.1) {$\c$};
  }
  \foreach \r/\c in {2/1,3/2,4/3,5/4} {
    \fill[blue!12] ({0.85+\r-1},1.9) rectangle ({0.85+\r},2.7);
    \node at ({0.35+\r},2.3) {$\c$};
  }
  \foreach \r/\c in {3/1,4/2,5/3,6/4} {
    \fill[blue!12] ({0.85+\r-1},1.1) rectangle ({0.85+\r},1.9);
    \node at ({0.35+\r},1.5) {$\c$};
  }
  \foreach \x in {0,...,7}
    \draw[black!45] ({0.85+\x},1.1) -- ({0.85+\x},3.5);
  \foreach \y in {0,...,3}
    \draw[black!45] (0.85,{1.1+0.8*\y}) -- (7.85,{1.1+0.8*\y});
  \draw[blue!65!black,semithick]
    (1.35,3.1) -- (2.35,2.3) -- (3.35,1.5);
  \draw[blue!65!black,semithick]
    (4.35,3.1) -- (5.35,2.3) -- (6.35,1.5);

  \draw[rounded corners=2pt,fill=orange!12,draw=black!60]
    (8.35,2.48) rectangle (11.25,3.82);
  \node[align=center,text width=2.05cm] at (9.8,3.15)
    {residual\\cohorts $5$--$7$\\$(p-1)f_P/T$};

  \node[draw=black!65,rounded corners=2pt,minimum width=1.25cm,
        minimum height=0.7cm,align=center] (spray) at (8.95,1.55)
    {spray\\$T+1$};
  \node[draw=black!65,rounded corners=2pt,minimum width=1.25cm,
        minimum height=0.7cm,align=center] (forward) at (10.65,1.55)
    {forward\\$T+2$};
  \draw[->,semithick] (spray) -- (forward);
  \draw[->,orange!80!black,semithick]
    (9.8,2.48) -- (9.8,2.25) -- (8.95,2.25) -- (spray.north);
\end{tikzpicture}
\end{minipage}
\caption{Cut-to-route schedule. (a) Pairing every zero-flow path with a reversed one-flow path gives combined key flow at most \(w_{ij}\) on arc \(i\to j\). (b) The first \(T-p+1\) cohorts are pipelined using public length \(w_{ij}/T\) per round; the remaining \(p-1\) cohorts use two-round spray-and-forward.}
\label{fig:01-cut-to-route-pipeline}
\end{figure}

\paragraph{Pipelining across rounds.}
For \(r\in[T]\), set \(\mathcal L_{ij}^r:=w_{ij}/T\). Split each paired path amount \(f_P\) into \(T\) equal cohorts. Launch the first \(T-p+1\) zero cohorts along \(P\), and their paired one cohorts along the reverse of \(P\), in rounds \(1,\ldots,T-p+1\); every launched cohort advances by one arc per round. Since every simple path has at most \(p-1\) arcs, all launched cohorts arrive by round \(T\) without violating temporal availability. In any round, their combined flow on \(i\to j\) is at most \((f^0_{ij}+f^1_{ij})/T\le w_{ij}/T=\mathcal L_{ij}^r\).
The first \(T\) rounds therefore form a feasible pipeline of cost at most \(C(w)\). In each pipeline round, the normalized receive load is $T^{-1}\max_j\sum_{i\ne j}w_{ij}\le pC(w)/T$.

\paragraph{Residual cleanup.}
The \(p-1\) unlaunched cohorts form an \(\eta:=(p-1)/T\) fraction of the paired transfer demand. Since \(\lvert b_i(z)\rvert\le1\), every processor has residual supply and demand at most \(\eta\). Split each residual source--destination transfer equally among all $p$ intermediate processors and then forward each share to its destination. Uniform splitting places at most \(\eta/p\) on every directed arc in either round, so set \(\mathcal L_{ij}^{T+1}=\mathcal L_{ij}^{T+2}:=\eta/p=(p-1)/(pT)\) for \(i\ne j\). These two rounds complete all residual transfers at additional cost \(2(p-1)^2/(pT)\).

Each cleanup round has normalized receive load $(p-1)^2/(pT)$. The cut separating the last $\lfloor p/2\rfloor$ all-zero processors from their disjoint sorted destinations and its complement give $C(w)\ge2\lfloor p/2\rfloor/p\ge(p-1)/p$. Hence the cleanup load is also at most $pC(w)/T$.
\end{proof}

Appending an empty round preserves feasibility, so \(C_{\mathrm{route}}(p,R)\) is nonincreasing in \(R\). Applying Lemma~\ref{lem:01-cut-to-route} to an optimal Cut-LP solution gives the matching asymptotic upper bound. Together with Theorem~\ref{lem:01-cut-necessity}, letting \(R\to\infty\) yields the following corollary.

\begin{corollary}[Cut/Route convergence]
\label{cor:01-lp-characterization}
For \(p\ge2\) and integer \(R\ge1\), \(C_{\mathrm{cut}}(p)\le C_{\mathrm{route}}(p,R)\). Moreover, if \(R\ge p+2\), then \(C_{\mathrm{route}}(p,R)\le C_{\mathrm{cut}}(p)+2(p-1)^2/(p(R-2))\). Consequently, \(C_{\mathrm{cut}}(p)=\inf_{R\ge1}C_{\mathrm{route}}(p,R)=\lim_{R\to\infty}C_{\mathrm{route}}(p,R)\).
\end{corollary}

Thus, even when a finite-round upper bound in Table~\ref{tab:01-results} does not match \(C_{\mathrm{cut}}(p)\), the optimal Route-LP cost approaches \(C_{\mathrm{cut}}(p)\) as \(R\) grows.

\subsubsection{Realization on indivisible items}
\label{sec:01-original-items}

It remains to realize the divisible upper bound on \(n\) indivisible items while preserving communication obliviousness. The following lemma compiles a finite-round Route-LP solution into a realistic algorithm. The time-expanded rounding argument is given in Appendix~\ref{app:01-integral-realization}.

\begin{lemma}[Indivisible realization]
\label{lem:01-route-lp-realization}
Let $\mathcal S$ be a feasible solution of $\mathrm{Route}(p,R)$. Then there is a deterministic perfectly communication-oblivious 0-1 sorting algorithm in the realistic model with $R+2$ rounds, communication $nC(\mathcal S)+O(Rp^2)$, and load at most $(n/p)\overline{\operatorname{load}}(\mathcal S)+O(Rp)$.
\end{lemma}

The algorithm and analysis are given in Appendix~\ref{app:01-generic-realization}.

Let \(C_{01}^{\mathrm{opt}}(p,n)\) be the minimum communication cost of deterministic perfectly communication-oblivious 0-1 sorting in the realistic model.

\begin{theorem}[Communication bound on indivisible items]
\label{thm:01-original-coefficient}
For every \(p\ge2\) and integer \(R\ge p+2\),
\[
nC_{\mathrm{cut}}(p)
\le C_{01}^{\mathrm{opt}}(p,n)
\le
\left(C_{\mathrm{cut}}(p)+\frac{2(p-1)^2}{p(R-2)}\right)n
+O(Rp^2).
\]
The upper bound is achieved by a deterministic perfectly communication-oblivious algorithm in the realistic model using $R+2$ rounds and load at most $\frac{nC_{\mathrm{cut}}(p)}{R-2}+O(Rp)$.
\end{theorem}

\begin{proof}
Theorem~\ref{lem:01-cut-necessity} gives the lower bound. For the upper bound, choose an optimal Cut-LP solution. Lemma~\ref{lem:01-cut-to-route} gives a feasible solution of $\mathrm{Route}(p,R)$ with objective value at most $C_{\mathrm{cut}}(p)+2(p-1)^2/(p(R-2))$ and normalized load at most $pC_{\mathrm{cut}}(p)/(R-2)$. Lemma~\ref{lem:01-route-lp-realization} then gives the stated bounds.
\end{proof}

Suppose $p=o(n^{1/3})$, and choose an integer $R\ge p+2$ such that $p=o(R)$ and $Rp^2=o(n)$; such a choice exists because $p^3=o(n)$. Together with the bound $C_{\mathrm{cut}}(p)\le3/2$ in Theorem~\ref{thm:01-cut-bounds}, Theorem~\ref{thm:01-original-coefficient} gives an algorithm using $R+2$ rounds, communication $C_{\mathrm{cut}}(p)n+O(pn/R+Rp^2)=C_{\mathrm{cut}}(p)n+o(n)$, and load $O(n/R+Rp)=o(n/p)$. Increasing $R$ reduces the temporal-routing surcharge $O(pn/R)$, whereas the indivisible-realization overhead $O(Rp^2)$ grows with $R$. For fixed $p$, any $R\to\infty$ with $R=o(n)$ is admissible, and consequently $C_{01}^{\mathrm{opt}}(p,n)/n\to C_{\mathrm{cut}}(p)$.

\subsection{Explicit bounds and separation from general sorting}
\label{sec:01-cut-coefficient}

Having established that $C_{\mathrm{cut}}(p)$ is the exact asymptotic communication coefficient for fixed $p$, and more generally when $p=o(n^{1/3})$, we now bound its value and compare it with the rigid-destination coefficient for general sorting.

\begin{theorem}[Explicit bounds on the Cut-LP coefficient]
\label{thm:01-cut-bounds}
For every \(p\ge2\), \(p^{-1}\lfloor(4p-2)/3\rfloor\le C_{\mathrm{cut}}(p)\le3/2\). For every \(p\ge3\), \(C_{\mathrm{cut}}(p)\le2-2/p-1/6\).
\end{theorem}

The lower and upper bounds follow from Lemmas~\ref{lem:01-cut-lower-bound} and~\ref{lem:01-cut-upper-bound}, respectively.

In particular, \(4/3-O(1/p)\le C_{\mathrm{cut}}(p)\le3/2\), while \((2-2/p)-C_{\mathrm{cut}}(p)\ge1/6\) for \(p\ge3\). Thus \(C_{\mathrm{cut}}(p)\) is uniformly smaller than the rigid-destination coefficient.

\section{Conclusion}
We establish tight leading-order communication bounds for communication-oblivious permutation, sorting, and 0-1 sorting when \(p=o(n^{1/3})\). The optimal costs are \((2-2/p)n+o(n)\) for permutation and general sorting, and \(C_{\mathrm{cut}}(p)n+o(n)\) for 0-1 sorting. Our constructive conversion from cut capacities to finite-round routes shows that the Cut-LP exactly characterizes the latter coefficient.
For every \(p\ge3\), we prove \(C_{\mathrm{cut}}(p)\le 2-2/p-1/6\), showing that the freedom to interchange equal-key items makes 0-1 sorting strictly cheaper than general sorting. Determining \(C_{\mathrm{cut}}(p)\) explicitly remains an open problem.

\bibliographystyle{plain}
\bibliography{main}

\appendix
\section{Proof of the Permutation Communication Lower Bound}
\label{app:perm-normalized-lower-bound}

The following is an adaptation of the direct-versus-detour counting argument in Keslassy et al.~\cite[Lemmas~13--15]{keslassy2005optimal}, rewritten for our normalized communication objective.

For \(s\in[p-1]\), let \(\pi_s(i):=((i-1+s)\bmod p)+1\), and let \(X^{(s)}\) be the balanced permutation input in which every item \(x\in X_i\) has key \(t_x=\pi_s(i)\). Define \(\cycf:=\{X^{(s)}:s\in[p-1]\}\). Thus \(|\cycf|=p-1\), all its inputs have the same leakage, and after normalization \(X^{(s)}\) induces one unit of demand from \(i\) to \(\pi_s(i)\).

\begin{proof}[Proof of Theorem~\ref{lb_cyclic_shift}]
Fix a normalized capacity matrix \(w\) as in the theorem. For each \(s\in[p-1]\), choose a fractional routing for \(X^{(s)}\), and let \(f_{ij}^{(s)}\) be the amount routed on arc \(i\to j\). Since \(f_{ij}^{(s)}\le w_{ij}\),
\begin{equation}\label{eq:perm-flow-upper}
\sum_{s=1}^{p-1}\sum_{i\ne j}f_{ij}^{(s)}
\le (p-1)\sum_{i,j}w_{ij}.
\end{equation}

Let \(f_s^{\mathrm{dir}}\) be the total mass routed directly on the prescribed arcs \(i\to\pi_s(i)\), summed over \(i\in[p]\). Since \(\pi_s\) has no fixed point, all \(p\) units must move. Each unit routed directly contributes one to the total arc flow, while every remaining unit contributes at least two. Hence
\begin{equation}\label{eq:perm-flow-lower}
\sum_{i\ne j}f_{ij}^{(s)}
\ge f_s^{\mathrm{dir}}+2\bigl(p-f_s^{\mathrm{dir}}\bigr)
=2p-f_s^{\mathrm{dir}}.
\end{equation}

For every off-diagonal arc \(i\to j\), there is a unique \(s\in[p-1]\) such that \(j=\pi_s(i)\). Thus this arc carries directly routed mass for at most one cyclic-shift demand, and
\begin{equation}\label{eq:perm-direct-upper}
\sum_{s=1}^{p-1}f_s^{\mathrm{dir}}
\le\sum_{i,j}w_{ij}.
\end{equation}

Combining the three inequalities gives \(p\sum_{i,j}w_{ij}\ge2p(p-1)\). Therefore \(C(w)=p^{-1}\sum_{i,j}w_{ij}\ge2-2/p\), as claimed.

The inputs in \(\cycf\) are valid indivisible balanced-permutation inputs, so the preliminary relaxation transfers the bound to indivisible algorithms.
\end{proof}

\section{Sparse cut and input families}
\label{app:01-sparse-support}

For every nonempty family \(\mathcal F\subseteq\mathcal F_{\mathrm{pure}}\), define \(\operatorname{dem}_{\mathcal F}(U):=\max_{z\in\mathcal F}\lvert b_z(U)\rvert\), and let \(C_{\mathrm{cut}}(p;\mathcal F)\) be the optimum of \(\mathrm{Cut}(p)\) after restricting the input family to \(\mathcal F\), equivalently after replacing \(\operatorname{dem}(U)\) by \(\operatorname{dem}_{\mathcal F}(U)\) in~\eqref{lp:01-cut}. Thus \(\operatorname{dem}_{\mathcal F_{\mathrm{pure}}}(U)=\operatorname{dem}(U)\).

\begin{lemma}[Sparse cut and input families]
\label{lem:01-sparse-hard-family}
For every \(p\ge2\), there exist a family \(\mathcal C_p\) of nontrivial cuts
and a family \(\mathcal F_p^\star\subseteq\mathcal F_{\mathrm{pure}}\) such that
\[
|\mathcal C_p|\le p^2,
\qquad
|\mathcal F_p^\star|\le p^2.
\]
Retaining only the rows indexed by \(\mathcal C_p\)
in~\eqref{lp:01-cut} preserves its optimal value, and
\(C_{\mathrm{cut}}(p;\mathcal F_p^\star)=C_{\mathrm{cut}}(p)\).
\end{lemma}

\begin{proof}
The Cut-LP has \(p^2\) capacity entries. After scaling the dual variables, its dual is
\[
\begin{array}{ll}
\text{maximize}
&\displaystyle
\frac1p\sum_{\emptyset\ne U\subsetneq[p]}\operatorname{dem}(U)y_U\\[1ex]
\text{subject to}
&\displaystyle
\sum_{\substack{\emptyset\ne U\subsetneq[p]\\i\in U,\ j\notin U}}y_U
\le1
\qquad(i,j\in[p]),\\[1ex]
&y_U\ge0
\qquad(\emptyset\ne U\subsetneq[p]).
\end{array}
\]

\paragraph{Reducing the dual support.}
Every nontrivial \(U\) contains some \(i\) and excludes some \(j\), so the corresponding constraint implies \(y_U\le1\). Thus the feasible region is bounded and has an optimal basic feasible solution \(y^*\). After adding slacks, the \(p^2\) inequalities above give \(p^2\) equations, so the standard basic-solution theorem implies that \(|\operatorname{supp}(y^*)|\le p^2\)~\cite{schrijver1986theory}.

\paragraph{Sparse cut and input families.}
Set \(\mathcal C_p:=\operatorname{supp}(y^*)\). The restriction of \(y^*\) is feasible for the
dual of the Cut-LP containing only these cut rows and retains objective
\(C_{\mathrm{cut}}(p)\). Weak duality and the fact that rows were only
deleted show that the reduced primal has the same optimal value.

For each \(U\in\mathcal C_p\), choose one vector in \(\mathcal F_{\mathrm{pure}}\) attaining \(\operatorname{dem}(U)\), and let \(\mathcal F_p^\star\) be the set of chosen vectors, discarding repetitions. Then
\(|\mathcal F_p^\star|\le p^2\) and \(\operatorname{dem}_{\mathcal F_p^\star}(U)=\operatorname{dem}(U)\) on \(\mathcal C_p\).
Consequently, \(y^*\) has the same objective in the
\(\mathcal F_p^\star\)-restricted dual. Weak duality gives
\(C_{\mathrm{cut}}(p;\mathcal F_p^\star)\ge C_{\mathrm{cut}}(p)\), while deleting
input constraints gives the reverse inequality.
\end{proof}

\section{Indivisible Realization of the 0-1 Sorting Algorithm}

\subsection{Integral realization of a divisible route}
\label{app:01-integral-realization}

We now transfer a fixed feasible divisible routing to the BSP model. Its public message lengths are normalized divisible amounts; we enlarge each round-arc length by fewer than two words and use flow integrality to route the items.

\begin{lemma}[Flow integrality with prescribed rounding]
\label{lem:flow_integrality}
Let a directed network have integral node demands and integral lower and upper arc bounds. If \(f\) is a feasible fractional flow, then there is an integral feasible flow \(\widehat f\) satisfying \(\lfloor f_e\rfloor\le\widehat f_e\le\lceil f_e\rceil\) for every arc \(e\).
\end{lemma}

\begin{proof}
Tighten each arc's lower and upper bounds to \(\lfloor f_e\rfloor\) and \(\lceil f_e\rceil\), respectively. The point \(f\) remains feasible for this bounded transshipment instance. Since the node-arc incidence matrix is totally unimodular, the corresponding flow polyhedron has an integral vertex whenever the node demands and arc bounds are integral~\cite{ahuja1993networkflows}. That vertex is the desired \(\widehat f\).
\end{proof}

\begin{lemma}[Indivisible routing for a fixed input]
\label{lem:01-indivisible-realization}
Fix an input \(X\) and normalized public lengths \((\mathcal L_{ij}^r)\). Suppose \(X\) admits a correct divisible routing with flows \(f_{ij}^{\sigma,z,r}\) and states \(m_i^{\sigma,z,r}\) under these lengths, where \(z:=z_X\). Then the integer lengths \(\widehat L_{ij}^r:=\lceil(n/p)\mathcal L_{ij}^r\rceil+1\) admit a correct routing of the items with communication at most
\[
\frac np\sum_{r=1}^R\sum_{i\ne j}\mathcal L_{ij}^r+2p(p-1)R.
\]
In each round $r$, processor $j$ receives at most $(n/p)\sum_{i\ne j}\mathcal L_{ij}^r+2(p-1)$ words.
\end{lemma}

\begin{proof}
Form a time-expanded network with a node \((i,r)\) for every \(i\in[p]\) and \(r\in[R+1]\). For each \(r\in[R]\), give communication arc \((i,r)\to(j,r+1)\) capacity \(\widehat L_{ij}^r\) for \(i\ne j\), and storage arc \((i,r)\to(i,r+1)\) capacity \(n\). Only communication arcs contribute to communication cost. The hat distinguishes these route-local lengths from the message lengths \(L_{ij}^r\) of Section~\ref{sec:preliminaries}.

\paragraph{Integral routing for each key class.}
Fix \(\sigma\in\{0,1\}\). In an independent copy of this network, assign supply \((n/p)m_i^{\sigma,z,1}\) to \((i,1)\) and demand \((n/p)m_i^{\sigma,z,R+1}\) to \((i,R+1)\). These quantities are integral because they are the initial and prescribed terminal item counts. Scaling the divisible routing by \(n/p\) gives a feasible fractional transshipment. Communication arc \((i,r)\to(j,r+1)\) carries \((n/p)f_{ij}^{\sigma,z,r}\), while storage arc \((i,r)\to(i,r+1)\) carries \((n/p)(m_i^{\sigma,z,r}-\sum_{j\ne i}f_{ij}^{\sigma,z,r})\).
The storage-arc value is nonnegative by (Avail), and (Flow) gives flow conservation. Applying Lemma~\ref{lem:flow_integrality} separately for \(\sigma=0,1\) yields integral flows with the prescribed endpoints, each arc rounded to its floor or ceiling.

\paragraph{Shared messages.}
The two integral flows fit into the common public messages. On every communication arc, their combined integral flow is at most
\[
\left\lceil \frac np f_{ij}^{0,z,r}\right\rceil
+\left\lceil \frac np f_{ij}^{1,z,r}\right\rceil
\le
\left\lceil \frac np\bigl(f_{ij}^{0,z,r}+f_{ij}^{1,z,r}\bigr)\right\rceil+1
\le
\left\lceil \frac np \mathcal L_{ij}^r\right\rceil+1
=\widehat L_{ij}^r.
\]
Thus, in every round, the public messages can carry the prescribed 0-items and 1-items. Equal-key items may be assigned arbitrarily to the integral flow, and flow conservation ensures that no item is copied or discarded. The terminal counts are exactly the sorted output.

Since \(\widehat L_{ij}^r\le(n/p)\mathcal L_{ij}^r+2\), summing over all \(p(p-1)R\) communication arcs proves the communication bound. Summing the same inequality over $i\ne j$ proves the receive-load bound.
\end{proof}

\subsection{Compiling a Route-LP solution into an oblivious algorithm}
\label{app:01-generic-realization}

Fix \(\mathcal S\) as in Lemma~\ref{lem:01-route-lp-realization} and use the rounded public lengths \(\widehat L_{ij}^r\) from Lemma~\ref{lem:01-indivisible-realization}. The algorithm has three phases.

\begin{enumerate}
\item \emph{Count collection and route construction.} In one fixed-length round, every other processor sends its local zero count to processor~1. From the resulting \(z_X\), the coordinator constructs a divisible route by Lemma~\ref{lem:lp_superposition} and integral flows by Lemma~\ref{lem:01-indivisible-realization}.

\item \emph{Instruction distribution.} In one fixed-length round, the coordinator sends each other processor a complete table of $2R(p-1)$ item counts, one for each round, outgoing destination, and key class. The entries follow a fixed public order and include zero counts.

\item \emph{Routing rounds.} In each route round \(r\in[R]\), which is global BSP round \(r+2\), processor \(i\) sends exactly \(\widehat L_{ij}^r\) words to every \(j\ne i\), following its instructions. Thus the notation of Section~\ref{sec:preliminaries} satisfies \(L_{ij}^{r+2}=\widehat L_{ij}^r\) during this phase.
\end{enumerate}

\paragraph{Analysis.}
Lemmas~\ref{lem:lp_superposition} and~\ref{lem:01-indivisible-realization} give an integral routing with the initial placement, exact sorted terminal counts, local availability in every round, and load at most $(n/p)\overline{\operatorname{load}}(\mathcal S)+2(p-1)$. Thus no item is copied or discarded. The two coordination rounds have input-independent lengths, the route rounds use the public lengths \(\widehat L_{ij}^r\), and fixed tie-breaking makes the construction deterministic and perfectly communication-oblivious.

The execution uses $R+2$ rounds and communicates $nC(\mathcal S)+O(Rp^2)$ words, by Lemma~\ref{lem:01-indivisible-realization} and the $O(Rp^2)$ coordination cost. Count collection delivers $p-1$ words to the coordinator, and instruction distribution delivers $2R(p-1)$ words to each other processor. Taking the maximum with the routing load gives load at most $(n/p)\overline{\operatorname{load}}(\mathcal S)+O(Rp)$, proving Lemma~\ref{lem:01-route-lp-realization}.

\section{Bounds for the Cut-LP}
\label{app:01-supplement}

\subsection{Cut-LP lower bound}
\label{app:01-forward}

\begin{lemma}[Cut-LP lower bound]
\label{lem:01-cut-lower-bound}
For every \(p\ge2\), \(C_{\mathrm{cut}}(p)\ge p^{-1}\lfloor(4p-2)/3\rfloor\).
\end{lemma}

\begin{proof}
For \(p=2\), the constraints for the two singleton cuts give \(w_{12},w_{21}\ge1\), and hence \(C(w)\ge1=p^{-1}\lfloor(4p-2)/3\rfloor\). Assume henceforth that \(p>2\).

Fix a feasible Cut-LP solution \(w\). Zeroing its diagonal preserves feasibility and can only decrease its cost. Choose an integer \(1\le k\le\lfloor p/2\rfloor\), and put \(U=[k]\). For each $k$-subset \(V\subseteq\overline U\), the pure input whose zero support is \(V\) has sorted zero target \(U\), and hence both \(w(V,\overline V)\ge k\) and \(w(\overline V,V)\ge k\). The cuts \(U\) and \(\overline U\) similarly give \(w(U,\overline U),w(\overline U,U)\ge k\).

Set \(w_{\mathrm{int}}:=w(\overline U,\overline U)\) and \(w_{\mathrm{cross}}:=\min\{w(U,\overline U),w(\overline U,U)\}\), so \(w_{\mathrm{cross}}\ge k\). Average both orientations of the preceding inequalities over all $k$-subsets \(V\subseteq\overline U\). In either orientation, a cross arc is counted in \(\binom{p-k-1}{k-1}\) cuts, and an arc internal to \(\overline U\) is counted in \(\binom{p-k-2}{k-1}\) cuts. Therefore
\[
\binom{p-k-1}{k-1}w_{\mathrm{cross}}+\binom{p-k-2}{k-1}w_{\mathrm{int}}
\ge
\binom{p-k}{k}k.
\]
Dividing by \(\binom{p-k-1}{k-1}\) gives \(w_{\mathrm{cross}}\ge p-k-(p-2k)w_{\mathrm{int}}/(p-k-1)\). Moreover, \(\sum_{i,j}w_{ij}\ge w_{\mathrm{int}}+w(U,\overline U)+w(\overline U,U)\ge w_{\mathrm{int}}+2w_{\mathrm{cross}}\). Combining the two lower bounds on \(w_{\mathrm{cross}}\) and minimizing over \(w_{\mathrm{int}}\ge0\) gives \(\sum_{i,j}w_{ij}\ge \min\{p+k-1,2(p-k)\}\).
The two values inside the minimum arise at the only relevant breakpoints: \(w_{\mathrm{int}}=0\), which gives \(2(p-k)\), and the point where the averaged lower bound on \(w_{\mathrm{cross}}\) meets \(w_{\mathrm{cross}}=k\), which gives \(p+k-1\). The same formula also covers \(p=2k\), when the averaged inequality is simply \(w_{\mathrm{cross}}\ge k\).

Finally, \(\max_{1\le k\le\lfloor p/2\rfloor}\min\{p+k-1,2(p-k)\}=\lfloor(4p-2)/3\rfloor\), so every feasible Cut-LP solution satisfies \(C(w)\ge\lfloor(4p-2)/3\rfloor/p\). Taking the minimum proves the bound on \(C_{\mathrm{cut}}(p)\).
\end{proof}

\subsection{Exact demand of a cut}

The ordered output has a particularly simple effect on cut demand. Recall that \(b_z(U)=\sum_{i\in U}b_i(z)\) and \(\operatorname{dem}(U)=\max_{z\in\mathcal F_{\mathrm{pure}}}\lvert b_z(U)\rvert\). For \(s\le p/2\), define the canonical left and right vertex subsets by \(V_s^{\mathrm L}:=[s]\) and \(V_s^{\mathrm R}:=\{p-s+1,\ldots,p\}\).

\begin{lemma}[Exact cut demand]
\label{lem:01-exact-cut-demand}
For every nontrivial \(U\subseteq[p]\) with \(s:=\lvert U\rvert\le p/2\),
\begin{equation}
\label{eq:01-exact-cut-demand}
\operatorname{dem}(U)
=s-\min\{\lvert U\cap V_s^{\mathrm L}\rvert,
           \lvert U\cap V_s^{\mathrm R}\rvert\}.
\end{equation}
In particular, \(s/2\le\operatorname{dem}(U)\le s\). Together with
\(\operatorname{dem}(U)=\operatorname{dem}(\overline U)\), this determines the demand of every nontrivial cut.
\end{lemma}

\begin{proof}
The pure input with zero support \(U\) has \(b_z(U)=s-\lvert U\cap V_s^{\mathrm L}\rvert\), whereas the one with zero support \(\overline U\) has \(b_z(U)=-(s-\lvert U\cap V_s^{\mathrm R}\rvert)\). Hence
\[
\operatorname{dem}(U)
\ge s-\min\{\lvert U\cap V_s^{\mathrm L}\rvert,
             \lvert U\cap V_s^{\mathrm R}\rvert\}.
\]
Conversely, for a pure input \(z\), put \(k:=\lVert z\rVert_1=\lvert\operatorname{supp}(z)\rvert\). Then \(b_z(U)=\lvert\operatorname{supp}(z)\cap U\rvert-\lvert U\cap[k]\rvert\). Its positive part is at most \(\min\{k,s\}-\lvert U\cap[k]\rvert\). This expression is nondecreasing for \(k\le s\) and nonincreasing for \(k\ge s\), so it is at most \(s-\lvert U\cap V_s^{\mathrm L}\rvert\). Similarly,
\[
-b_z(U)
\le \lvert U\cap[k]\rvert-\max\{0,k-(p-s)\},
\]
which is nondecreasing for \(k\le p-s\) and nonincreasing for \(k\ge p-s\), and is therefore at most \(\lvert U\cap[p-s]\rvert=s-\lvert U\cap V_s^{\mathrm R}\rvert\). This proves~\eqref{eq:01-exact-cut-demand}.

Finally, \(V_s^{\mathrm L}\) and \(V_s^{\mathrm R}\) are disjoint, so \(\lvert U\cap V_s^{\mathrm L}\rvert+\lvert U\cap V_s^{\mathrm R}\rvert\le s\), giving \(s/2\le\operatorname{dem}(U)\le s\).
\end{proof}

\subsection{Cut-LP upper bound}
\label{app:01-cut-bounds}

\begin{lemma}[Cut-LP upper bound]
\label{lem:01-cut-upper-bound}
For every $p\ge2$, $C_{\mathrm{cut}}(p)\le3/2$.
Moreover, for every $p\ge3$,
$C_{\mathrm{cut}}(p)\le2-2/p-1/6$.
\end{lemma}

\begin{proof}
\textbf{Construction.}
Partition $[0,3)$ into the unit intervals
$I_a^{\mathrm{unit}}=[a-1,a)$ for $a\in[3]$, and represent
processor $i$ by $I_i^{\mathrm{proc}}=[3(i-1)/p,3i/p)$.
Put $\mu_{i,a}:=|I_i^{\mathrm{proc}}\cap I_a^{\mathrm{unit}}|$.
Set $w_{ii}=0$, and for $i\ne j$ define
\[
w_{ij}:=\frac{3p}{16}
\sum_{a,b\in[3]}|a-b|\mu_{i,a}\mu_{j,b}.
\]
These capacities are nonnegative and symmetric.

\par\smallskip\noindent\textbf{Objective value.}
The processor intervals partition $[0,3)$, so
$\sum_i\mu_{i,a}=1$ for every $a\in[3]$.
Also, $\sum_{a,b\in[3]}|a-b|=8$.
Including the nonnegative diagonal terms in the objective
would therefore give $3/2$.
Since those terms are omitted, $C(w)\le3/2$.

\par\smallskip\noindent\textbf{Cut constraints.}
Next we verify the constraint for every nontrivial cut $U$.
By symmetry of $w$ and
$\operatorname{dem}(U)=\operatorname{dem}(\overline U)$,
we may take its cut size $s:=|U|$ to satisfy $s\le p/2$.
Let $I_U:=\bigcup_{i\in U}I_i^{\mathrm{proc}}$.
For $a\in[3]$, define the mass of this cut in the
$a$-th unit interval by
\[
\mu_{U,a}:=\sum_{i\in U}\mu_{i,a}
          =|I_U\cap I_a^{\mathrm{unit}}|.
\]
The cut $U$ remains fixed below, so we abbreviate
$\mu_{U,a}$ to $\mu_a$.

Write $\bar s:=3s/p$ for the normalized cut size;
it is also the length of $I_U$.
Set $\bar w(U):=3w(U,\overline U)/p$ and
$\overline{\operatorname{dem}}(U)
:=3\operatorname{dem}(U)/p$.
Thus $\mu_1+\mu_2+\mu_3=\bar s\le3/2$.

Because the complement of $U$ occupies length
$1-\mu_b$ in $I_b^{\mathrm{unit}}$, summing the
capacities across the cut gives
\[
\begin{aligned}
\bar w(U)
&=\frac9{16}\sum_{a,b\in[3]}
  |a-b|\mu_a(1-\mu_b)\\
&=\frac9{16}\bigl(
  2\bar s+2\mu_1^2+2\mu_3^2
  +(1-2\bar s)(\mu_1+\mu_3)
  \bigr).
\end{aligned}
\]

The first and last $s$ processor intervals form
$[0,\bar s)$ and $[3-\bar s,3)$.
By Lemma~\ref{lem:01-exact-cut-demand},
\[
\overline{\operatorname{dem}}(U)
=\bar s-\min\bigl\{
  |I_U\cap[0,\bar s)|,\,
  |I_U\cap[3-\bar s,3)|
  \bigr\}.
\]
Reflection of the processor order preserves both
the construction and the demand.
We may therefore assume that the left intersection
is the smaller one.

If $\bar s\le1$, at most $1-\bar s$ of the first
unit interval lies outside $[0,\bar s)$.
If $\bar s\ge1$, this prefix contains the entire
first unit interval, while at most $2-\bar s$
of the second lies outside it.
Consequently, using $\mu_2=\bar s-\mu_1-\mu_3$,
we obtain
\[
\overline{\operatorname{dem}}(U)\le
\begin{cases}
\min\{\bar s,1-\mu_1\},&\bar s\le1,\\
\min\{\bar s-\mu_1,\,2-\bar s+\mu_3\},&\bar s\ge1.
\end{cases}
\]

\par\smallskip\noindent\textbf{Feasibility.}
It remains to show
$\bar w(U)\ge\overline{\operatorname{dem}}(U)$.
Put $\xi_s:=(2\bar s-1)/4$ and consider
the following ranges of $\bar s$.

\begin{enumerate}
\item \emph{$0<\bar s\le1$.}
  \begin{enumerate}
  \item Suppose $\mu_1\le1-\bar s$.
    The applicable demand bound is $\bar s$, and
    completing squares gives
    \[
    \bar w(U)-\bar s
    =\frac{-36\bar s^2+44\bar s-9}{64}
     +\frac98\bigl(
       (\mu_1-\xi_s)^2+(\mu_3-\xi_s)^2
       \bigr).
    \]
    If $\bar s\le1/2$, then $\xi_s\le0$,
    and this expression is at least $\bar s/8$.
    For $1/2\le\bar s\le5/6$, its first term
    is concave and positive at both endpoints.
    For $5/6\le\bar s\le1$, the constraint
    $\mu_1\le1-\bar s\le\xi_s$ gives the lower bound
    \[
    \frac{252\bar s^2-452\bar s+207}{128}>0.
    \]
    The numerator has positive leading coefficient
    and negative discriminant.

  \item Suppose $\mu_1\ge1-\bar s$.
    The applicable demand bound is $1-\mu_1$.
    The quadratic $\bar w(U)-(1-\mu_1)$ has center
    $(\xi_s-4/9,\xi_s)$, whose first coordinate
    is smaller than $1-\bar s$.
    Its minimum over the half-plane
    $\mu_1\ge1-\bar s$ occurs at
    $(1-\bar s,\xi_s)$ and gives the same
    positive lower bound as above.
  \end{enumerate}

\item \emph{$1\le\bar s\le7/6$.}
  The demand bound is at most $\bar s-\mu_1$.
  The quadratic
  $\bar w(U)-(\bar s-\mu_1)$ has center
  $(\xi_s-4/9,\xi_s)$.
  Since $\xi_s<4/9$ and $\mu_1\ge0$,
  completing squares gives
  \[
  \bar w(U)-(\bar s-\mu_1)
  \ge\frac{-36\bar s^2+52\bar s-9}{128}>0.
  \]
  The numerator is concave and positive
  at both endpoints of this range.

\item \emph{$7/6\le\bar s\le3/2$.}
  Put $v:=\mu_1+\mu_3$, the selected mass
  in the two outer unit intervals.
  Direct expansion gives
  \[
  \begin{aligned}
  &\bar w(U)
   -\min\{\bar s-\mu_1,\,2-\bar s+\mu_3\}
   -\frac1{72}\\
  &\qquad=
   \frac9{16}
     \left(\mu_3-\mu_1-\frac49\right)^2
   +\frac9{16}(1-v)(2\bar s-2-v)
   +\frac12|2\bar s-2-v|.
  \end{aligned}
  \]
  \begin{enumerate}
  \item If $v\le2\bar s-2$, every term on
    the right is nonnegative, since
    $v\le2\bar s-2\le1$.

  \item If $v\ge2\bar s-2$, the last two
    terms combine into
    $(v-2\bar s+2)(9v-1)/16$.
    This is nonnegative because
    $v\ge2\bar s-2\ge1/3$.
  \end{enumerate}
\end{enumerate}

Thus $\bar w(U)\ge
\overline{\operatorname{dem}}(U)$ for every
nontrivial cut $U$. Hence $w$ is Cut-LP feasible,
and $C_{\mathrm{cut}}(p)\le C(w)\le3/2$.

\par\smallskip\noindent\textbf{Uniform gap.}
For $p\ge6$, the construction gives
$(2-2/p)-C_{\mathrm{cut}}(p)
\ge1/2-2/p\ge1/6$.
For $p=3,4,5$, the certified values
$1,1,6/5$ in Table~\ref{tab:01-results}
give the same claimed gap.
\end{proof}

\end{document}